\documentclass[a4paper]{article}
\usepackage{amsmath,amsfonts,amssymb,amsthm}
\usepackage{a4wide}
\newcommand{\R}{{\mathbb R}}

\newcommand{\p}{{\bf p}}

\newcommand{\y}{{\bf y}}

\newcommand{\x}{{\bf x}}

\newcommand{\bk}{{\bf k}}
\numberwithin{equation}{section}
\theoremstyle{plain}
\newtheorem{theorem}{Theorem}[section]

\newtheorem{lemma}[theorem]{Lemma}
\begin{document}

\noindent 
\begin{center}
\textbf{\large Faraday effect revisited: sum rules and convergence issues}
\end{center}

\begin{center}
March 29, 2010
\end{center}

\vspace{0.5cm}

\noindent 

\begin{center}
\textbf{ 
Horia D. Cornean\footnote{Department of Mathematical Sciences, 
    Aalborg
    University, Fredrik Bajers Vej 7G, 9220 Aalborg, Denmark; e-mail:
    cornean@math.aau.dk} and Gheorghe Nenciu\footnote{ Inst. of Math. ``Simion Stoilow'' of
the Romanian Academy, P. O. Box 1-764, RO-014700 Bucharest, Romania, and Department of Mathematical Sciences, 
    Aalborg
    University, Fredrik Bajers Vej 7G, 9220 Aalborg, Denmark;
e-mail: Gheorghe.Nenciu@imar.ro}
}
     
\end{center}

\vspace{0.5cm}

\noindent

\begin{abstract}
This is the third paper of a series revisiting the Faraday
effect. The question of the absolute convergence of the sums over the band indices entering the Verdet constant 
is considered. In general, 
sum rules and traces per unit volume play an important role in 
solid state physics, and they give rise to certain convergence
problems widely ignored by physicists. We give a complete answer in
the case of smooth potentials and formulate an open problem
related to less regular perturbations. 
\end{abstract}

\section{Introduction}

In this paper, which is the third in a series, we continue our study begun in \cite{CNP} where we outlined a general and mathematically sound theory of the Faraday effect. The theoretical investigation of this effect -which amounts to the study of the transverse conductivity  of the optical sample- has a long and distinguished history in the quantum theory of solids. While sending the reader to \cite{CNP} for a detailed discussion and more references, we only mention here that the first general theory of the conductivity tensor for independent electrons subjected to a periodic electric potential and to a constant magnetic field in the linear response approximation has been developed by  L. Roth \cite{R1}. Roth's method is based on an effective Hamiltonian approach for Bloch electrons in a weak magnetic field (see \cite{R2} and references therein). But due to the difficulties generated by the linear growth at infinity of the magnetic vector potential, the theory is highly formal. Basic questions as for example computation of higher order corrections in the magnetic field strength, the thermodynamic limit as well as some convergence issues, seem to be almost hopeless to cope with. Accordingly, the obtained formulae -although considered as a landmark of the subject- still contain singular terms at the crossings of the Bloch bands and only met a moderate success at the practical level.

In the first paper of this series we outlined a rigorous theory of the Faraday effect free of all these difficulties. The first basic idea is to employ a method going back at least to Sondheimer and Wilson \cite{SoWi} in order to express the traces involved in computing various physical quantities  as integrals using 
Green's functions, which are easier to control than the eigenfunctions and eigenvalues. The second and the main tool is the use of gauge covariance via a regularized magnetic perturbation theory \cite{BC, Cor1, CN, CN22, CN2, IMP, MP1, MP2, MP3, adn}. The result is a well defined, gauge invariant, workable formula for the transverse conductivity tensor in the thermodynamic limit. This tensor has an asymptotic expansion in the strength of the magnetic field, whose first two coefficients have been explicitly written down in terms of the zero field Green function, which in turn can be expressed in terms of zero field Bloch energies and eigenfunctions (see Section VII in \cite{CNP}). 

The subsequent papers of our series are devoted to some more delicate mathematical aspects of the program left open in \cite{CNP}.
In particular,  in the second paper of the series \cite{CN2} we solved -by making use of the full power of the regularized magnetic perturbation theory- the most difficult part of the program: the proof of the existence of the thermodynamic limit.

The present paper is devoted to another convergence issue left open in \cite{CNP}. More precisely, when expressing the relevant quantities in terms of Bloch eigenfunctions and eigenvalues one arrives, by formal computations, at expressions containing multiple sums over the band indices as well as integrals over Brillouin and unit cells. The point is that the order of summation is crucial since the sums might not be absolutely convergent. The problem is more than an academic one since in practical computations finite bands models are used which amounts to truncate the infinite sums so that the question whether the sums are absolutely convergent is a crucial one. The main result of this paper is that the smoothness of the periodic potential $V$ guarantees that such series are always absolutely convergent. The smoothness of $V$ is not superfluous as shown by an example. Since the same kind of  series appear very oftenly in solid state physics (e.g. when computing derivatives of dispersion laws) and are related to sum rules concerning various observables, we present our main technical result  Theorem \ref{teorema1} in a general setting .

The content of the paper is as follows.  Section 2 contains the setting, examples of series we are interested in as well as the  related sum rules. In Section 3 we give our main technical result giving the decay (under suitable smoothness condition of $V$) of the matrix elements of the momentum operator in the Bloch basis as the band index increases. Section 4 contains applications to series appearing in the theory of the Faraday effect. Finally in Section 5 we present a computation involving the (singular)  one dimensional $\delta$-"potential", in order to point out the need of smoothness conditions on $V$. The question to optimally relate the decay properties of the momentum matrix elements in the Bloch basis to the degree of smoothness of $V$ is left as an interesting open problem.

\section{Preliminaries and motivation} \label{sectiuneadoi}

The one particle Hilbert
space for a $d$-dimensional spinless particle is
$
\mathcal{H}_{\infty}=L^2(\R^d)$,  $d\geq 1$.
The Hamiltonian describing its dynamics is:
\begin{equation}\label{hamimare}
H_\infty={\bf P}^2 +V,\quad {\bf P}=-i\nabla,
\end{equation}
where we assume that $V$ is 
$\epsilon$-relatively bounded with respect to $-\Delta$,
and periodic with respect to a Bravais lattice $\mathcal{L}$. 
Denote by $\Omega$ its elementary cell and by 
$\Omega^*$ the corresponding Brillouin zone.  

We assume for simplicity that the lattice is
$\mathbb{Z}^d$, $d\geq 1$, and the unit cell $\Omega$ is the unit
cube centred at the origin. The Bloch-Floquet unitary is 
(see e.g. \cite{RS4}): 
\begin{align}\label{adoua1}
&U:L^2(\R^d)\mapsto \int_{\Omega^*}^\oplus L^2(\Omega)d\bk, \nonumber \\
&(Uf)(\underline{x},\bk):=\frac{1}{(2\pi)^{d/2}}\sum_{\gamma\in
  \mathbb{Z}^d} e^{-i\bk\cdot (\underline{x}+\gamma)}f(\underline{x}+\gamma),\quad \bk\in
  \Omega^*,\quad \underline{x}\in \Omega,\quad f\in C_0^\infty(\R^d).
\end{align}
Then $H_\infty$ can be written as a fiber integral 
$H_\infty=\int_{\Omega^*}^\oplus h(\bk)d\bk$ where the fiber operator 
\begin{align} \label{fibHa}
h(\bk) &= \left (-i\nabla_p +\bk\right )^2
+V,
\end{align}
is defined in  $L^2(\Omega)$ with periodic boundary
conditions, and its domain is the periodic 
Sobolev space $H^2(\Omega_p)$. Here $-i\nabla_p$ denotes the momentum
operator in $L^2(\Omega)$ with periodic boundary conditions. The
operator $h(\bk)$ has compact resolvent and purely
discrete spectrum $\{\lambda_j(\bk)\}_{j\geq 1}$. We can choose a set
of eigenfunctions $\{u_j(\cdot,\bk)\}_{j\geq 1}$ 
which form an orthonormal basis of $L^2(\Omega)$ and obey the equation: 
\begin{equation}
h(\bk)u_j(\cdot,\bk)=\lambda_{j}(\bk)u_j(\cdot,\bk).
\end{equation}
We label $\lambda_{j}(\bk)$ in increasing order. We
have to remember that due to this choice, $\lambda_{j}(\bk)$ and
 $u_j(\x,\bk)$ are not differentiable with respect to 
$\bk$ at crossing points. Without loss of generality, we may also
assume that $h(\bk)\geq 1$ for all $\bk$ which implies:
\begin{equation}\label{sept1}
\lambda_{j}(\bk)\geq 1,\quad \bk\in
  \Omega^*,\quad j\geq 1.
\end{equation}

One can introduce the Bloch functions:
\begin{equation}\label{blochf}
\Psi_j(\x,\bk):=\frac{e^{i \bk\cdot\x}}{(2\pi)^{\frac{d}{2}}}
u_j(\x,\bk),\quad \x\in\R^d
\end{equation}
which form  a basis of generalized eigenfunctions of $H_\infty$ in $L^2(\R^d)$,
i.e. in distributional sense we have:
$$\int_{\Omega^*}\sum_{j\geq 1}\Psi_j(\x,\bk)\overline{\Psi_j}(\y,\bk)d\bk =\delta(\x-\y).$$ 
Using Bloch functions, one can express the integral kernel of any
function of $H_\infty$ in terms of the fiber $h(\bk)$. For example, the Green 
function (i.e. the integral kernel of the resolvent) writes as:
 \begin{equation}\label{intker}
G_\infty(\x,\y; z)=\int_{\Omega^*}\sum_{j\geq 1}
\frac{\Psi_j(\x,\bk) \overline{\Psi_{j}(\y,\bk)} }
{\lambda_j(\bk)-z}d\bk, 
\end{equation}
The above formula has to be understood in the distributional sense 
since the series on
the right hand side is typically not absolutely convergent. 
Now let us present the problems we want to consider.

\subsection{Perturbation theory and sum rules}

 By writing 
$$h(\bk)=h(\bk_0)+2(\bk-\bk_0)\cdot [-i\nabla_p
+\bk_0]+(\bk-\bk_0)^2=:h(\bk_0)+W,$$
we see that $W$ is a regular perturbation for $h(\bk_0)$. Assume that
$\lambda_1(\bk_0)$ is non-degenerate. Then according to the analytic
perturbation theory, $\lambda_1(\bk)$ remains non-degenerate in a neighborhood of $\bk_0$. Keeping
$|\bk-\bk_0|$ 
small enough and using the Feshbach formula with the projection 
$\Pi_0=|u_1(\cdot,\bk_0)\rangle \langle u_1(\cdot,\bk_0)|$ we obtain:
\begin{align}\label{prima3}
\lambda(\bk)&=\lambda(\bk_0)+\langle u_1(\cdot,\bk_0),W u_1(\cdot,\bk_0)\rangle\nonumber \\ 
&-\langle W u_1(\cdot,\bk_0),\left \{\Pi_0^{\perp} 
[h(\bk)-\lambda(\bk)]\Pi_0^{\perp}\right\}^{-1}W u_1(\cdot,\bk_0)\rangle.
\end{align}
By iterating the above formula we can identify in an efficient way the full Taylor
expansion of $\lambda_1$ around $\bk_0$. Let us first introduce the notation:
\begin{align}\label{kerl3}
\hat{\pi}_{ij}(\alpha,\bk):=\langle u_i(\cdot,\bk),(-i\partial_\alpha+
k_\alpha)u_j(\cdot,\bk)\rangle_{L^2(\Omega)} =\int_{\Omega}u_i(\x,\bk) 
\overline{[(-i\partial_\alpha+ k_\alpha) u_j]}(\x,\bk) d\x.
\end{align}

If we are interested for
example in $\frac{\partial^4\lambda_1}{\partial k_\alpha^4}$, one of
the terms building it will be proportional with: 
\begin{equation}\label{apatra2}
\sum_{j_1\geq 2}\hat{\pi}_{1j_1}(\alpha,\bk)\sum_{j_2\geq 2}\frac{\hat{\pi}_{j_1j_2}(\alpha,\bk)}{\lambda_{j_2}-\lambda_1}
\sum_{j_3\geq
  2}\frac{\hat{\pi}_{j_2j_3}(\alpha,\bk)}{\lambda_{j_3}-\lambda_1}\hat{\pi}_{j_31}(\alpha,\bk),
\end{equation}
where the order of the sums is crucial. Each series generates 
an $l^2(\mathbb{N})$-summable vector for the next one. The multiple series is
convergent, but not apriori absolutely convergent. This is because the 
inner $\hat{\pi}$'s can grow in absolute value with the energy. But we
will prove in 
this paper that such series are always absolutely
convergent if $V$ is smooth. 

Before giving some more precise results, let us make the connection
with the notion of sum rules and argue why the above matrix elements
may have a rapid decay with the energy if one index is kept fixed. We
work in $L^2(\Omega)$ and let us omit $\bk$ in order to simplify
notation. If $A$ is an operator which is relatively bounded to $h$,
then we can write the 
identity ($t\geq 0$):
\begin{equation}\label{apatra1}
 i\langle [A,e^{ith}Ae^{-ith}]u_m,u_m\rangle =2\sum_{n\geq 1}|\langle
 Au_m,u_n\rangle|^2 
\sin\{t(\lambda_m-\lambda_n)\}. 
\end{equation}
The commutator on the left hand side is well defined as a quadratic
form on functions belonging to the domain of $h$. Now if we assume that 
$hAh^{-N}$ is bounded for some $N\geq 1$, then we obtain the bound 
$$  |\langle Au_m,u_n\rangle|= \frac{\lambda_m^{N}}{\lambda_n}
|\langle hAh^{-N}u_m,u_n\rangle|$$
which when applied to \eqref{apatra1} would insure that the series
defines a 
differentiable function at all $t$. 
If in addition we assume that we can make sense out of the double
commutator $[A,[h,A]]$, 
then a typical sum rule would be:
\begin{equation}\label{sumrule}
 \langle [A,[h,A]]u_m,u_m\rangle =2\sum_{n\geq 1}|\langle
 Au_m,u_n\rangle|^2 \cdot (\lambda_n-\lambda_m).
\end{equation} 
If $A$ does not commute well with $h$, it can happen that 
even if the function of 
$t$ in \eqref{apatra1} is continuous on $\R$, we cannot be sure that
it is also differentiable. 
We will later on give an example of a Schr\"odinger operator which
will make appear the Riemann function on the right hand side of
\eqref{sumrule}. This function is known 
to be everywhere continuous, but differentiable only at certain 
rational points not including $t=0$.

For many more rigorous aspects and applications of sum rules, see \cite{HS}.  

\subsection{Traces per unit volume}

Let us consider the Fermi-Dirac distribution function:
\begin{equation}\label{fermidira}
f_{FD}(x):=\frac{1}{e^{\beta(x-\mu)}+1},\quad x\in \R, \beta >0,\mu\in
\R, 
\end{equation}
and choose the following path in the complex plane
\begin{equation}\label{contur1}
\Gamma=\left \{x\pm i\delta :\;-1\leq x
  <\infty
\right \}\bigcup 
\left \{-1+i y:\:-\delta\leq y \leq
\delta \right \} 
\end{equation}
with a fixed $\delta\in (0,\frac{\pi}{2\beta}]$. Introduce the notation:
\begin{equation}\label{prima1}
R_\infty(z):=(H_\infty-z)^{-1}.
\end{equation}

The quantity we are interested in this time is a trace-per-unit-volume
(provided it exists):
\begin{align}\label{stone5}
I_{\alpha_1,\dots , \alpha_n}:=\frac{1}{|\Omega|}\int_{\Omega}d\x\left \{
\int_{\Gamma}dz
{f}_{FD}(z) 
P_{\alpha_1}R_\infty(z)P_{\alpha_2}R_\infty(z) \dots P_{\alpha_n}R_\infty(z)\right \}(\x,\x).
\end{align}
We want to express the above quantity only with the help of Bloch
functions and energies. The fact that the operator 
\begin{equation}\label{adoua3}
\int_{\Gamma} 
\left \{P_{\alpha_1}R_\infty(z)P_{\alpha_2}R_\infty(z) \dots
P_{\alpha_n}R_\infty(z)\right \}{f}_{FD}(z) dz
\end{equation}
has a jointly continuous integral kernel if $V$ is smooth, 
can be proved with the same
methods as in for example \cite{CN2}, even in the presence of constant
magnetic fields and without any periodicity condition on $V$. 
In the periodic case the situation is somewhat simpler, as we will 
see in the next section.

Now let us use \eqref{intker} and \eqref{blochf} in \eqref{stone5} and 
perform formal computations using the completeness of Bloch functions
and freely interchanging  the order of various integrals and series. 
We would arrive at an expression as follows: 

\begin{align}\label{prima2}
&I_{\alpha_1,\dots , \alpha_n} =\frac{1}{(2\pi)^d}\int_{\Omega^*}d\bk \\
&\left \{\sum_{j_1,\dots,j_n\geq 1}\hat{\pi}_{j_1j_2}(\alpha_1,\bk)\dots \hat{\pi}_{j_nj_1}(\alpha_n,\bk)
\int_{\Gamma}
 \frac{{f}_{FD}(z)}{(\lambda_{j_1}(\bk)-z)\dots (\lambda_{j_n}(\bk)-z)}dz\right\}.\nonumber 
\end{align}
The first question is why is \eqref{prima2}
convergent? The answer is not obvious. Moreover, if we look at 
$\hat{\pi}_{ij}(\alpha,\bk)$ alone (assume without loss that
$\lambda_j\geq \lambda_i$), we see that using Cauchy's inequality we
obtain that it could grow like $||(-i\partial_\alpha+ k_\alpha)
u_i||\sim \sqrt{\lambda_i}$, which is not very encouraging. For 
some more physical background on 
such problems, see \cite{PC}.

\section{A technical result}

\begin{theorem}\label{teorema1} Let $M\geq 1$ be an integer and assume
  that $V\in C^{2M-1}(\R^d)$ is a $\Omega$-periodic potential with
  $V\geq 1$. Then 
for every integer $0\leq N\leq M$ there exists a constant $C_N$ such that
uniformly in $\bk\in \Omega^*$, $\alpha\in\{1,\dots,d\}$ and $s,t\geq 1$ we have the estimate:
\begin{equation}\label{prima4}
 |\hat{\pi}_{st}(\alpha,\bk)|=|\hat{\pi}_{ts}(\alpha,\bk)|\leq
 C_N\frac{\lambda_s^{N+\frac{1}{2}}}{\lambda_t^{N}}.
\end{equation}
\end{theorem}

\begin{proof}
During this proof we drop the $\bk$
dependence of the various quantities. We denote
$-i\partial_\alpha$ with $p_\alpha$. 
The eigenfunctions $u_j$'s are a-priori in $H^2(\Omega_p)$ and
$\Omega$-periodic functions. The lower bound on $V$ insures that
$h\geq 1$. The idea of the proof is based on
the following identity:
\begin{equation}\label{prima5}
 \hat{\pi}_{ts}(\alpha)\frac{\lambda_t^{N}}{\lambda_s^{N}}=\langle
h^{N}u_t, (p_\alpha +k_\alpha)h^{-N}u_s\rangle\hat{\pi}_{ts}(\alpha)+\langle u_t, h^N[p_\alpha,h^{-N}]u_s\rangle,
\end{equation}
assuming for the moment that the commutator on the right maps $L^2$ into the domain of $h^N$. 

We know that $p_\alpha h^{-1/2}$ is bounded, thus we can write: 
\begin{equation}\label{prima6}
 |\hat{\pi}_{ts}(\alpha)|\leq C \lambda_s^{1/2}.
\end{equation}
Hence \eqref{prima4} would follow if we can prove the estimate
$$||h^N[p_\alpha,h^{-N}]||\leq C_N.$$

\subsection{An induction argument}
\begin{lemma}\label{lema1}
Choose any $\phi\in C^{2M}(\Omega_p)$. Then  
the operators $ h^N[p_\alpha,h^{-N}]$ and $h^N\phi h^{-N} $ 
are bounded for any $1\leq N\leq M$.
\end{lemma}

\begin{proof}
We give a proof by induction in $M$. Let us start with $N=M=1$. Then we have
$$h \phi h^{-1}=\phi +\{-2i (\nabla\phi)\cdot (\p +\bk)-(\Delta \phi)\}h^{-1}$$
and 
\begin{equation}\label{prima7}
h[p_\alpha,h^{-1}]=[h,p_\alpha]h^{-1}=i(\partial_\alpha V) h^{-1}
\end{equation}
and both are clearly bounded. Note the important thing that we only lost two derivatives of $\phi$ in the first relation. 

Now take $M\geq 2$. The induction assumption is the following: for any $N\leq M-1$, and for any function $\tilde{\phi}\in C^{2M-2}(\Omega_p)$ we have 
\begin{equation}\label{asasea1}
||h^{N}[p_\alpha,h^{-N}]||\leq C_N,\quad ||h^N\tilde{\phi} h^{-N}||\leq C_{N,\tilde{\phi}},\quad 1\leq N\leq M-1.
\end{equation}
Note that the induction hypothesis allows us to conclude that 
\begin{equation}\label{prima8}
h^N[\phi,h^{-N}]=h^N\phi h^{-N}-\phi,\quad N\leq M-1,
\end{equation}
is also bounded. 

We now have to prove that we can allow $N=M$. 
The main idea is to rewrite both $h^{M}[p_\alpha,h^{-M}]$ and $ h^{M}\phi
h^{-M} $ as linear combinations of similar terms but of lower order. We are compelled 
to perform the induction step simultaneously, because the reduction
step mixes the two types of operators.  

We start with an identity:
\begin{equation}\label{prima9}
[p_\alpha,h^{-M}]=h^{-1}[p_\alpha,h^{-M+1}]+[p_\alpha,h^{-1}]h^{-M+1}=h^{-1}[p_\alpha,h^{-M+1}]+i h^{-1}(\partial_\alpha V) h^{-M},
\end{equation}
where in the second inequality we used \eqref{prima7}. Multiplying on
the left by $h^{M}$ we obtain:
\begin{equation}\label{prima10}
h^{M}[p_\alpha,h^{-M}]=h^{M-1}[p_\alpha,h^{-M+1}]+i h^{M-1}(\partial_\alpha V) h^{-M},
\end{equation}
Here the main point is that $\partial_\alpha V$ belongs to $C^{2M-2}(\Omega_p)$, 
thus $h^{M-1}(\partial_\alpha V) h^{-M+1}$ must be bounded according to \eqref{asasea1}. 

Now let us deal with the other type of operator. We write another
identity:
\begin{align}\label{prima100}
\phi h^{-M}&=h^{-M}\phi +[\phi,h^{-M}]=h^{-M}\phi
+h^{-M+1}[\phi,h^{-1}]+ [\phi,h^{-M+1}]h^{-1}\nonumber \\
&=h^{-M}\phi +h^{-M+1}[\phi,h^{-1}]+h^{-1}[\phi,h^{-M+1}]-[h^{-1},[\phi,h^{-M+1}]].
\end{align}
The first three operators in the last equality have their images in the domain of $h^{M}$ due to the
induction hypothesis \eqref{prima8}. For the fourth term we apply the
Jacobi identity for the double commutator and obtain:
\begin{align}\label{prima11}
[h^{-1},[\phi,h^{-M+1}]]&=[h^{-M+1},[\phi,h^{-1}]]=[h^{-M+1},h^{-1}[h,\phi]h^{-1}]\nonumber
\\
&=h^{-1}[h^{-M+1},[h,\phi]]h^{-1}.
\end{align}
The commutator $[h,\phi]$ can be easily put in the form 
$\sum_{\beta=1}^df_\beta p_\beta +g$, with $f_\beta$'s proportional with $\partial_\beta \phi$, while $g$ is containing up to second order partial derivatives of $\phi$. Thus:
\begin{align}\label{prima12}
& h^{-1}[h^{-M+1},[h,\phi]]h^{-1}=h^{-1}\sum_{\beta=1}^d[h^{-M+1},f_\beta
p_\beta]h^{-1}+h^{-1}[h^{-M+1},g]h^{-1}\nonumber \\
&=h^{-M}\left \{\sum_{\beta=1}^dh^{M-1}[h^{-M+1},f_\beta]
p_\beta h^{-1}+\sum_{\beta=1}^d(h^{M-1}f_\beta h^{-M+1})(h^{M-1}[h^{-M+1},p_\beta])h^{-1}\right\} \nonumber \\
&+h^{-M}\left \{h^{M-1}[h^{-M+1},g]h^{-1}\right \}.
\end{align}
The induction hypothesis will insure that the operator inside the large parenthesis is bounded, because our 
$f_\beta$'s and $g$ belong in particular to $C^{2M-2}$. Thus according
to \eqref{asasea1} and \eqref{prima8}we have that $h^{M-1}f_\beta h^{-M+1}$ and $h^{M-1}[h^{-M+1},g]$ are bounded operators. This concludes the proof.
\end{proof}

\noindent{\bf Remark}. Let us show that \eqref{prima4} is sharp in the
free case ($V=0$) and for $t=s$. (Choose the lattice to be
$\mathbb{Z}^3$). We have (due to the Feynman-Hellman identity) that 
$2\hat{\pi}_{tt}(\alpha,\bk)=\partial_\alpha \lambda_t(\bk)$ at any
point $\bk$ where $\lambda_t(\bk)$ is non-degenerate, and
$\lambda_t=(2\pi{\bf m}+\bk)^2$ for some ${\bf m}\in
\mathbb{Z}^3$. Then $|\partial_\alpha \lambda_t(\bk)|\sim |{\bf
  m}|\sim \sqrt{\lambda_t}$ for large $t$. 

The theorem is definitely not sharp when $V=0$ and $t\neq s$, because in that case 
all off-diagonal matrix elements are identically zero. 
\end{proof}

\section{Applications}

In this section we solve the problems outlined in the
introduction, and start by considering traces per unit volume. In order to simplify the arguments, we choose $V\in C^\infty(\Omega_p)$. 

Assume that $A(\bk)$ is a family of bounded, integral operators on
$L^2(\Omega)$, with a jointly continuous integral kernel $
\mathcal{A}(\x,\x';\bk)$, kernel which can be extended to a jointly continuous and periodic
function on $\R^{2d}$. The operator $B:=U^*\; \int_{\Omega^*}^\oplus
A(\bk)d\bk \; U$ defined on $L^2(\R^d)$ has an integral kernel
$\mathcal{B}(\x,\x')$ given by (see e.g. \cite{BCZ}):
\begin{equation}\label{adoua2}
\mathcal{B}(\x,\x')=\frac{1}{(2\pi)^{d}}\int_{\Omega^*}
e^{i\bk\cdot (\x-\x')}\mathcal{A}(\x,\x';\bk) d\bk.
\end{equation}
We are interested in a $B$ operator as given in \eqref{adoua3}. Its
corresponding $A(\cdot)$ fiber family is:
\begin{align}\label{adoua4}
A(\bk)=\int_{\Gamma}dz
{f}_{FD}(z) 
(p_{\alpha_1}+k_{\alpha_1})(h(\bk)-z)^{-1} \dots
(p_{\alpha_n}+k_{\alpha_n})(h(\bk)-z)^{-1}. 
\end{align}

Now if we can prove that $A(\bk)$ has a jointly continuous integral
kernel, then 
\begin{equation}\label{adoua5}
I_{\alpha_1\dots
  \alpha_n}=\frac{1}{(2\pi)^{d}}\int_{\Omega^*}\int_{\Omega}\mathcal{A}(\x,\x;\bk) d\x \;d\bk.
\end{equation}

We can integrate by parts $M$ times with respect to $z$ in \eqref{adoua4}  using an exponentially decaying antiderivative $f_M$ 
of $f_{FD}$, and write $A(\bk)$ as a sum of operators of the following type:

\begin{align}\label{adoua6}
A_{m_1\dots m_n}(\bk)&:=\int_{\Gamma}dz
{f}_{M}(z) 
(p_{\alpha_1}+k_{\alpha_1})(h(\bk)-z)^{-m_1} \dots
(p_{\alpha_n}+k_{\alpha_n})(h(\bk)-z)^{-m_n}, \\
& m_1+\dots+m_n=n+M. \nonumber
\end{align}

Without loss, we may assume that $m_1\geq M/n$, and keep in mind that
$M$ 
will be chosen to be a large number. The integrand becomes trace class 
if $m_1-1$ is larger than some number $m_d\geq 1 $ depending on the
dimension $d$. The integral kernel of 
$A_{m_1\dots m_n}(\bk)$ can be written in distributional sense as:

\begin{align}\label{adoua7}
&\mathcal{A}_{m_1\dots m_n}(\x,\x';\bk)\\\nonumber  
&:=\int_{\Gamma}dz
{f}_{M}(z) 
\sum_{j_1\geq 1}u_{j_1}(\x,\bk)\frac{\hat{\pi}_{j_1j_2}(\alpha_1,\bk)}{(\lambda_{j_2}(\bk)-z)^{m_1}} \dots
\sum_{j_{n+1}\geq 1}\frac{\hat{\pi}_{j_{n}j_{n+1}}(\alpha_n,\bk)}{(\lambda_{j_{n+1}}(\bk)-z)^{m_n}} u_{j_{n+1}}(\x',\bk). 
\end{align}
If we can prove that the above multiple series and the integral are
absolutely convergent uniformly in $\x$, $\x'$ and $\bk$, then the
integral kernel would be a continuous function as a uniform limit of continuous functions. 

Due to standard results on elliptic regularity and Sobolev embedding
we can prove the existence of $\tilde{m}_d\geq 1$ such that 
\begin{align}\label{adoua8}
|u_{j}(\x,\bk)|\leq C ||h(\bk)^{\tilde{m}_d}u_{j}(\cdot,\bk)||=C \lambda_j(\bk)^{\tilde{m}_d},
\end{align}
uniformly in $j$ and $\bk$. 
Thus we can bound \eqref{adoua7} with:
\begin{align}\label{adoua9}
&|\mathcal{A}_{m_1\dots m_n}(\x,\x';\bk)|\\
&\leq C \int_{\Gamma}dz
|{f}_{M}(z)|
\sum_{j_1,\dots,j_{n+1}\geq 1}\lambda_{j_1}(\bk)^{\tilde{m}_d}\lambda_{j_{n+1}}(\bk)^{\tilde{m}_d}\frac{|\hat{\pi}_{j_1j_2}(\alpha_1,\bk)|}
{|\lambda_{j_2}(\bk)-z|^{m_1}} \dots \frac{|\hat{\pi}_{j_{n}j_{n+1}}(\alpha_n,\bk)|}{|\lambda_{j_{n+1}}(\bk)-z|^{m_n}}, \nonumber 
\end{align}
and we will now prove that the above quantity is finite. First, note
that $1/|\lambda_{j}(\bk)-z|$ is uniformly bounded in $j\geq 1$ and
$z\in \Gamma$. Second, using the triangle inequality and the previous
remark we have 
$$\frac{\lambda_{j}(\bk)^{m}}
{|\lambda_{j}(\bk)-z|^{m}} \leq C_{m} (1+|z|^2)^{m/2}$$
again uniformly in $j\geq 1$ and $z\in \Gamma$. 
Third, using the exponential decay
of $|f_M|$ and the previous estimate we have: 
\begin{align}\label{adoua10}
|{f}_{M}(z)| \frac{1}
{|\lambda_{j_2}(\bk)-z|^{m_1}} \leq C_{m_1}\lambda_{j_2}(\bk)^{-m_1}(1+|z|)^{-2}.
\end{align}
Replacing this estimate in the integral we obtain:
\begin{align}\label{adoua11}
&|\mathcal{A}_{m_1\dots m_n}(\x,\x';\bk)|\\\nonumber 
&\leq C \sum_{j_1,\dots,j_{n+1}\geq  1}\lambda_{j_2}(\bk)^{-m_1}
\lambda_{j_1}(\bk)^{\tilde{m}_d}\lambda_{j_{n+1}}(\bk)^{\tilde{m}_d}|\hat{\pi}_{j_1j_2}(\alpha_1,\bk)|
\dots
|\hat{\pi}_{j_{n}j_{n+1}}(\alpha_n,\bk)|. 
\end{align}
As mentioned before, there exists $m_d$ large enough such that 
$\sum_{j\geq 1}\lambda_j(\bk)^{-m_d}<\infty$ uniformly in 
$\bk$. Now it is time to use \eqref{prima4} in a certain recursive
way. First, we get rid of $j_1$ by using the estimate:
\begin{align}\label{adoua16}
|\hat{\pi}_{j_{1}j_{2}}(\alpha_{1},\bk)|&\leq C 
\frac{\lambda_{j_{2}}(\bk)^{m_d+\tilde{m}_d+1/2}}{\lambda_{j_{1}}(\bk)^{m_d+\tilde{m}_d}}.
\end{align}
Replacing this in \eqref{adoua11} will provide us with a convergent
series in $j_1$, the price being the apparition of a higher power in
$\lambda_{j_2}$. 

We can now get rid of all other indices different from $j_2$, starting
with $j_{n+1}$ and going down to $j_2$. For instance we write:
\begin{align}\label{adoua12}
|\hat{\pi}_{j_{n}j_{n+1}}(\alpha_n,\bk)|&\leq C 
\frac{\lambda_{j_{n}}(\bk)^{m_d+\tilde{m}_d+1/2}}{\lambda_{j_{n+1}}(\bk)^{m_d+\tilde{m}_d}},
\end{align}
which replaced in \eqref{adoua11} allows us to sum over $j_{n+1}$ and obtain the estimate:
\begin{align}\label{adoua13}
|\mathcal{A}_{m_1\dots m_n}(\x,\x';\bk)|&\leq C
\sum_{j_2,\dots,j_{n}\geq  1}
\lambda_{j_2}(\bk)^{-m_1+m_d+\tilde{m}_d+1/2}
\dots
\lambda_{j_n}(\bk)^{m_d+\tilde{m}_d+1/2}|\hat{\pi}_{j_{n-1}j_{n}}(\alpha_n,\bk)|. 
\end{align}
Then we can get rid of $j_n$ by writing:
\begin{align}\label{adoua14}
|\hat{\pi}_{j_{n-1}j_{n}}(\alpha_n,\bk)|&\leq C 
\frac{\lambda_{j_{n-1}}(\bk)^{r_1+1/2}}{\lambda_{j_{n}}(\bk)^{r_1}},\quad r_1\geq 2m_d+\tilde{m}_d+1/2,
\end{align}
and then summing. Next comes $j_{n-1}$: 
\begin{align}\label{adoua15}
|\hat{\pi}_{j_{n-2}j_{n-1}}(\alpha_{n-1},\bk)|&\leq C 
\frac{\lambda_{j_{n-2}}(\bk)^{r_2+1/2}}{\lambda_{j_{n-1}}(\bk)^{r_2}},\quad r_2\geq m_d+r_1+1/2,
\end{align}
and so on. Thus we can recursively go down to $j_2$, generating a 
positive but finite power of $\lambda_{j_2}$. In the end we make use
of $m_1$ (in fact $M$) in order to make the series in $j_2$ convergent. 

We have therefore shown that the trace class operator in \eqref{adoua6} has a jointly continuous integral kernel. Moreover, 
we now can apply Fubini's theorem and express its trace as: 

\begin{align}\label{atreia1}
&\frac{1}{(2\pi)^d}\int_{\Omega^*}d\bk \\
&\left \{\sum_{j_1,\dots,j_n\geq 1}\hat{\pi}_{j_1j_2}(\alpha_1,\bk)\dots \hat{\pi}_{j_nj_1}(\alpha_n,\bk)
\int_{\Gamma}
 \frac{{f}_{M}(z)}{(\lambda_{j_1}(\bk)-z)^{m_1}\dots (\lambda_{j_n}(\bk)-z)^{m_n}}dz\right\}.\nonumber 
\end{align}

Now let us observe that the $z$ integral can be computed through
residue calculus, and each resulting term will contain a factor of the
type 
$f_M^{(t)}(\lambda_{j_s})$ with $1\leq t\leq M$ and $1\leq s\leq
n$. Each such factor decays exponentially. 
This allows us to prove that each such term converges absolutely, by
using the exponential decay of $f_M$ in order to control the growth of the
momentum matrix elements. As this point we can put together all the
$z$ integrals coming from all operators as the one in
\eqref{adoua6}. Then integrals with respect to $z$ will add up and  
form the integral 
$$\int_{\Gamma}
 \frac{{f}_{FD}(z)}{(\lambda_{j_1}(\bk)-z)\dots
   (\lambda_{j_n}(\bk)-z)}dz $$
which will still generate exponentially decaying factors like
$f_{FD}^{(t)}(\lambda_{j_s})$, where $1\leq t\leq n$ and $1\leq s\leq n$. 
Thus we obtain the final formula announced in \eqref{prima2}.  

Now let us come back to sum rules and derivatives of
eigenvalues. While the convergence of \eqref{sumrule} is a direct consequence of
\eqref{prima4}, the absolute convergence of \eqref{apatra2} can be
proved with the same trade-off method as the one we used for the trace
problem. 

\section{An open problem and an example}

An interesting open problem is to relate the decay of our matrix
elements with the regularity of $V$. According to Theorem \ref{teorema1}, if $V$ is 
not infinitely differentiable but has a number of continuous
derivatives, then one can still prove some decay in energy for the
matrix elements. The question is whether one can
describe the exact decay rate of such a matrix element in 
case of singular potentials. 

Let us investigate a singular ``potential'' for which one can give an
exact asymptotic
expression for $\hat{\pi}_{1j}$ when $j\to \infty$. Let us consider
the one dimensional operator $h:=-\frac{d^2}{dx^2}+g\delta (x)$, $g>0$,
defined through its quadratic form in 
$L^2([-\frac{1}{2},\frac{1}{2}])$, with periodic boundary
conditions. Its domain includes the periodic Sobolev space
$H^1((-\frac{1}{2},\frac{1}{2}]_p)$ where the respective interval is
identified with the $1d$-torus.

The operator $h$ commutes with the inversion operator. 
The delta potential does not act on the odd subspace, 
and the odd eigenfunctions 
are $u_j(x)=\sqrt{2} \sin(2j\pi x)$, $j\geq 1$, corresponding to
eigenvalues $\lambda_j=4\pi^2 j^2$. On the even subspace, the
eigenfunctions must be of the form:
\begin{equation}\label{acincea1}
\tilde{u}_j(x)=C_j\{\cos(\beta_jx)+\frac{g}{\beta_j}\sin(\beta_jx)\},\quad
j\geq 1,\quad x\in [0,1/2],
\end{equation}
where $C_j$ is a normalization constant and
$\tilde{u}_j(x):=\tilde{u}_j(-x)$ for $x\in (-1/2,0)$. Such a function
cannot belong to $H^2((-\frac{1}{2},\frac{1}{2}]_p\setminus\{0\})$ unless
$\tilde{u}'_j\left (\frac{1}{2}_-\right )=0$. This provides a quantization 
condition for $\beta_j>0$,
which must solve the equation 
\begin{equation}\label{acincea2}
\beta \tan(\beta/2)=g,\quad \beta\in (2\pi(j-1),2\pi j),\quad j\geq 1.
\end{equation}
The differences $\beta_j-2\pi(j-1)$ and $C_j-\sqrt{2}$ go to zero when $j\to\infty$. 

Now let us consider the scalar product $\hat{\pi}_j:=i\langle \tilde{u}_1,
u_j'\rangle=-i\langle \tilde{u}_1',
u_j\rangle$ and investigate its behavior when $j$ grows. Due to
symmetry considerations we have:
\begin{align}\label{sept2}
\hat{\pi}_j &=-i2\sqrt{2}
j\int_0^{1/2}\tilde{u}_1'(x)\sin(2j\pi
x)dx=-i\frac{\sqrt{2}}{j\pi}\tilde{u}_1'(0_+)-i\frac{\sqrt{2}}{j\pi}
\int_0^{1/2}\tilde{u}_1''(x)\cos(2j\pi
x)dx\nonumber \\
&=i\frac{C_1 g \sqrt{2}}{j\pi} +i\frac{\sqrt{2}}{2j^2\pi^2}
\int_0^{1/2}\tilde{u}_1'''(x)\sin(2j\pi
x)dx\nonumber \\
&=i\frac{C_1 g \sqrt{2}}{j\pi}
+i\frac{\sqrt{2}}{4j^3\pi^3}\{\tilde{u}_1'''(0_+)-\tilde{u}_1'''(1/2)\cos(\pi
j)\}+\mathcal{O}(1/j^3)\nonumber \\
&=i\frac{C_1 g \sqrt{2}}{j\pi}+\mathcal{O}(1/j^3).
\end{align}

Now if we replace in \eqref{apatra1} the operator $A$ with $p:=-id/dx$,
$u_m$ with $\tilde{u}_1$, and $u_n$ with $u_j$, we obtain: 
\begin{equation}\label{acincea3}
 i\langle [p,e^{ith}pe^{-ith}]\tilde{u}_1,\tilde{u}_1\rangle =C\sum_{j\geq 1}\frac{\sin\{t4\pi^2j^2\}}{j^2} +F(t), 
\end{equation}
where $C$ is a nonzero constant and $F(t)$ is some differentiable
function at $t=0$. On the right hand side we recognize the Riemann
function, which in particular is not differentiable at $t=0$. Therefore
a sum rule like in \eqref{sumrule} cannot hold. Nevertheless, the
function in \eqref{acincea3} is H\"older continuous at $t=0$ of order
$\alpha\in [0,1/2)$.

\section{Acknowledgments}

This paper is dedicated to our late friend Pierre Duclos. 
Part of this work was carried out while G.N. was visiting professor at 
Department of Mathematical Sciences, Aalborg University. 
H.C. acknowledges support from the Danish FNU grant 
{\it Mathematical Physics}.


\begin{thebibliography}{}





\bibitem{BC}  {Ph. Briet, H.D. Cornean.} {  Locating the spectrum for
 magnetic Schr\"odinger and Dirac operators.}
{\it Comm. Partial
 Differential Equations} {\bf 27} (5-6), 1079--1101 (2002)




\bibitem{BCZ}  Briet, P., Cornean, H.D., Zagrebnov, V.A.: 
"Do Bosons Condense in a Homogeneous Magnetic
Field?", Journal of Statistical Physics {\bf 116} (5/6), 1545-1578 (2004)

\bibitem{Cor1} Cornean, H.D.:
\newblock 
On the magnetization of a charged Bose gas in the canonical ensemble. 
\newblock {\it Commun. Math. Phys.} 212 (1), 1-27  (2000)



\bibitem{CN} Cornean, H.D., Nenciu, G.: 
\newblock On eigenfunction decay for two dimensional magnetic Schr\"odinger operators.
\newblock {\it Commun. Math. Phys.} {\bf 192}, 671-685 (1998)



\bibitem{CN22} Cornean, H.D., Nenciu, G.: 
\newblock  Two-dimensional magnetic
  Schr\"odinger operators: width of mini bands in the tight binding approximation.
 Ann. Henri Poincar{\'e}  {\bf 1} (2)  (2000), p.~203-222.



\bibitem{CNP}  Cornean, H.D., Nenciu, G., Pedersen, T. G.: 
The Faraday effect revisited: general theory.  {\it J. Math. Phys.}  
{\bf 47} (1), 013511 (2006)

\bibitem{CN2} Cornean, H.D., and Nenciu, G.: Faraday rotation
  revisited: Thermodynamic limit.  {\it J. Funct. Anal.}  {\bf 257} (7), 2024-2066 (2009) 


\bibitem{HS} Harrell, E. M., II; Stubbe, J.: 
Trace identities for commutators, with applications to the distribution of eigenvalues. {\it Trans. Amer. Math. Soc.}, 
in press. 


\bibitem{IMP}Iftimie, V., M\u antoiu, M., Purice, R.: 
Magnetic Pseudodifferential Operators. {\it Publications of the Research Institute for Mathematical Sciences} 
{\bf 43} (3), 585-623 (2007)




\bibitem{MP1}M\u antoiu, M., Purice, R.: 
Strict deformation quantization for a particle in a magnetic field.  
 J. Math. Phys.  {\bf 46}  (5), 052105 (2005) 	


\bibitem{MP2}M\u antoiu, M., Purice, R.: The magnetic Weyl calculus. 
 J. Math. Phys. {\bf 45} (4), 1394--1417 (2004)

\bibitem{MP3} M\u antoiu, M., Purice, R., Richard, S.: Spectral and propagation results for magnetic 
Schrodinger operators; 
A $C^*$-algebraic framework. {\it J. Funct. Anal.} {\bf 250} (1), 42-67 (2007)



\bibitem{adn} { Nenciu,G.:} { On asymptotic perturbation theory for
 quantum mechanics:
almost invariant subspaces and gauge invariant magnetic perturbation
 theory.} {\it J. Math. Phys.},
{\bf 43} (3), 1273--1298 (2002)


\bibitem{PC} Pedersen, T.G., Cornean, H.D.: 
Optical second harmonic generation from Wannier excitons.  {\it
  Europhysics Letters}  {\bf 78}, 27005 (2007) 


 \bibitem {RS4} {Reed, M., Simon, B::} {\it Methods of Modern
  Mathematical Physics IV: Analysis of Operators},  Academic Press, 1978.



\bibitem{R1}
 Roth, L.M.: "Theory of Bloch electrons in a magnetic field", 
J.Phys. Chem. Solids {\bf 23}, 433 (1962).

\bibitem{R2} Roth, L.M.: "Theory of Faraday effect in solids", 
Phys. Rev. {\bf 133}, A542 (1964).

\bibitem{SoWi} Sondheimer, E.H., Wilson, A.H.: The theory of the magneto-resistance effects in metals. {\it Proc. Roy. Soc. A}  {\bf 190} (1023), 435-455  (1947) 


  









\end{thebibliography}
\end{document}